\documentclass[conference,10pt]{IEEEtran}
\usepackage{pgfplots} 
\usepackage{comment}
\usepackage{amsfonts,enumitem}
\usepackage{algorithmic}
\usepackage{textcomp}
\def\BibTeX{{\rm B\kern-.05em{\sc i\kern-.025em b}\kern-.08em
		T\kern-.1667em\lower.7ex\hbox{E}\kern-.125emX}}
\usepackage{amsmath,amssymb,multirow}
\usepackage{setspace}
\usepackage{listings}
\usepackage{cite}
\usepackage{stackengine}
\usepackage{acronym}
\usepackage{mathtools}
\usepackage{subfigure}
\usepackage{algorithm,algorithmic}
\DeclareMathOperator*{\argmax}{arg\,max}
\DeclareMathOperator*{\argmin}{arg\,min}
\newcommand\abs[1]{\left|#1\right|}
\newcommand\norm[1]{\left\lVert #1\right\rVert}
\newcommand{\bp}{\boldsymbol{p}}
\newcommand{\bcc}{\boldsymbol{c}}
\newcommand{\by}{\boldsymbol{y}}

\newcommand{\bmu}{\boldsymbol{\mu}}
\newcommand{\tmu}{\tilde{{\mu}}}
\newcommand{\tbmu}{\tilde{\boldsymbol{\mu}}}
\newcommand{\balp}{{\boldsymbol {\alpha}}}
\newcommand{\bbalp}{\bar{\boldsymbol {\alpha}}}
\newcommand{\bbp}{\bar{\boldsymbol{p}}}
\newcommand{\bb}{\boldsymbol{b}}
\newcommand{\bw}{\boldsymbol{w}}
\newcommand{\bbet}{\bar{\bet}}
\newcommand{\bet}{\boldsymbol{\eta}}
\newcommand{\bA}{\boldsymbol{A}}
\newcommand{\bB}{\boldsymbol{B}}
\newcommand{\bep}{\boldsymbol{\epsilon}}
\newcommand{\rmx}{\mathrm{x}}
\newcommand{\rmy}{\mathrm{y}}
\newcommand{\rmz}{\mathrm{z}}
\newcommand{\ba}{\boldsymbol{a}}
\newcommand{\bu}{\boldsymbol{u}}
\newcommand{\projrange}[1]{\boldsymbol{\Pi}_{#1}}
\usepackage{amsthm}
\theoremstyle{plain}
\newtheorem{thm}{Theorem}
\theoremstyle{remark}
\newtheorem{rem}[thm]{Remark}
\theoremstyle{plain}

\newtheorem{lemma}{Lemma}

\graphicspath{{./Figures/}}

\acrodef{RIS}{reconfigurable intelligent surface}
\acrodef{SNR}{signal-to-noise ratio}
\acrodef{ISAC}{integrated sensing and communication}
\acrodef{ISLAC}{integrated sensing, localization, and communication}
\acrodef{LoS}{line-of-sight}
\acrodef{AoA}{angle-of-arrival}
\acrodef{AoD}{angle-of-departure}
\acrodef{UE}{user equipment}
\acrodef{BS}{base station}
\acrodef{MCRB}{misspecified Cram\'{e}r-Rao bound}
\acrodef{LB}{lower bound}
\acrodef{MML}{mismatched maximum likelihood}

\pgfplotsset{
	tick label style={font=\small},
	label style={font=\small},
	legend style={font=\small}
}

\usepackage{ellipsis}
\usetikzlibrary{calc}
\usetikzlibrary{decorations.pathreplacing,decorations.markings,shapes.geometric}
\tikzset{naming/.style={align=center,font=\small}}
\tikzset{antenna/.style={insert path={-- coordinate (ant#1) ++(0,0.25) -- +(135:0.25) + (0,0) -- +(45:0.25)}}}
\tikzset{station/.style={naming,draw,shape=dart,shape border rotate=90, minimum width=10mm, minimum height=10mm,outer sep=0pt,inner sep=3pt}}
\tikzset{mobile/.style={naming,shape=circle,draw, minimum width=4mm,minimum height=1mm, outer sep=0pt,inner sep=3pt}}
\tikzset{risblockage/.style={naming,fill = red!30, very thick, shape=rectangle,minimum width=5mm,minimum height=10mm, outer sep=0pt,inner sep=3pt}}
\tikzset{radiation/.style={{decorate,decoration={expanding waves,angle=90,segment length=4pt}}}}

\newcommand{\UE}[1]{%
\begin{tikzpicture}[every node/.append style={circle,minimum width=0pt}]
\node[mobile] (box) {#1};
\end{tikzpicture}
}

\newcommand{\RISBlockage}[1]{%
	\begin{tikzpicture}
		\node[risblockage] (box) {#1};
	\end{tikzpicture}
}

\newcommand{\BS}[1]{%
\begin{tikzpicture}
\node[station] (base) {#1};

\draw[line join=bevel] (base.100) -- (base.80) -- (base.110) -- (base.70) -- (base.north west) -- (base.north east);
\draw[line join=bevel] (base.100) -- (base.70) (base.110) -- (base.north east);

\draw[line cap=rect] ([yshift=0pt]base.north) [antenna=1];
\end{tikzpicture}
}

\newcommand{\RIS}[1]{%
	\begin{tikzpicture}
		\draw[step=0.125,thick, draw = blue] (0,0) grid (1.5,1.5);
	\end{tikzpicture}
}
\begin{document}
	\bstctlcite{IEEEexample:BSTcontrol}
	
	\title{On the Impact of Hardware Impairments on RIS-aided Localization}
	\author{C\"{u}neyd \"{O}zt\"{u}rk\IEEEauthorrefmark{1}, Musa Furkan Keskin\IEEEauthorrefmark{2}, Henk Wymeersch\IEEEauthorrefmark{2}, Sinan Gezici\IEEEauthorrefmark{1}\\
	\IEEEauthorrefmark{1}Department of Electrical
and Electronics Engineering, Bilkent University, Turkey\\
\IEEEauthorrefmark{2}Department of Electrical
 Engineering, Chalmers University of Technology, Sweden}
	\maketitle
	\begin{abstract}
	We investigate a reconfigurable intelligent surface (RIS)-aided near-field localization system with single-antenna user equipment (UE) and base station (BS) under hardware impairments by considering a practical phase-dependent RIS amplitude variations model. To analyze the localization performance under the mismatch between the practical model and the ideal model with unit-amplitude RIS elements, we employ the misspecified Cram\'{e}r-Rao bound (MCRB). Based on the MCRB derivation, the lower bound (LB) on the mean-squared error for estimation of UE position is evaluated and shown to converge to the MCRB at low signal-to-noise ratios (SNRs). Simulation results indicate more severe performance degradation due to the model misspecification with increasing SNR. In addition, the mismatched maximum likelihood (MML) estimator is derived and found to be tight to the LB in the high SNR regime. Finally, we observe that the model mismatch can lead to an order-of-magnitude localization performance loss at high SNRs.
	\end{abstract}

	\section{Introduction}\label{sec:Intro}
	Among the envisioned technological enablers for 6G, three stand out as being truly disruptive: the transition from 30 GHz to beyond 100 GHz (the so-called higher mmWave and lower THz bands) \cite{saad2019vision,rappaport2019wireless,tataria20216g}, the convergence of communication, localization, and sensing (referred to as \ac{ISAC} or \ac{ISLAC}) \cite{chiriyath2017radar,liu2021integrated,de2021convergent,wymeersch2021integration,JCS_2021}, and the introduction of \acp{RIS} \cite{RIS_tutorial_2021,elzanaty20216g,wymeersch2020radio}. RISs are large passive surfaces, comprising arrays of metamaterials, and have the ability to shape the propagation environment, thus locally boosting the \ac{SNR} to improve communication quality \cite{RIS_Access_2019,RIS_commag_2021}. This is especially relevant in beyond 100 GHz to overcome sudden drops in rate caused by temporary blockage of the \ac{LoS} path. 
	
	In parallel with the benefits for communications, \acp{RIS} can similarly improve localization performance \cite{RIS_2018_TSP}. Stronger even, \acp{RIS} with known position and orientation have the ability to enable localization in scenarios where it would otherwise be impossible \cite{Keykhosravi2020_SisoRIS}. In this respect, the large aperture of the \ac{RIS} has several interesting properties. First of all, the \ac{SNR}-boosting provides accurate delay measurements when wideband signals are used \cite{RIS_bounds_TSP_2021,rappaport2019wireless}. Secondly, the large number of elements provides high resolution in \ac{AoA} (for uplink localization) or \ac{AoD} (for downlink localization) \cite{rappaport2019wireless}. Third, when the \ac{UE} is close to the \ac{RIS} (in the sense that the distance to the \ac{RIS} is of similar order as the physical size of the \ac{RIS}), wavefront curvature effects (so-called geometric near-field) can be harnessed to localize the user \cite{RIS_2018_TSP,Shaban2021,EM_wavefront,rahal2021risenabled}, even when the \ac{LoS} path between the \ac{BS} and \ac{UE} is blocked, irrespective of whether wideband or narrowband signals are used. 
	\begin{figure}
	\centering
	    \begin{tikzpicture}[every path/.append style={thick}]
            \matrix[column sep=.75cm,row sep=.75cm]
            {
            & \node(C) [label={below:{\textbf{RIS}}}]{\RIS{}}; &  &  \\
            \node[every path/.append style={thick},inner sep=0pt](A){\BS{\textbf{BS}}}; & \node(D){\RISBlockage{LoS \\Blockage}}; &   &\node[thick](B){\UE{\textbf{UE}}}; \\
            };
            \draw[thick,radiation,decoration={angle=45}] (A.north) -- +(45:0.5);
            \draw[thick] (A.north)--(C.center);
            \draw[thick] (C.center)--([xshift = -.1768cm, yshift = -.1768cm]B.north);
            \draw[thick,radiation,decoration={angle=45}] ([xshift = -.6cm, yshift = .2768cm]B.north)-- + (-45:0.5);
        \end{tikzpicture}
        \caption{Configuration of a RIS-aided localization system with LoS blockage.}
        \vspace{-0.5cm}
         \label{fig:RISconfiguration}
	\end{figure}
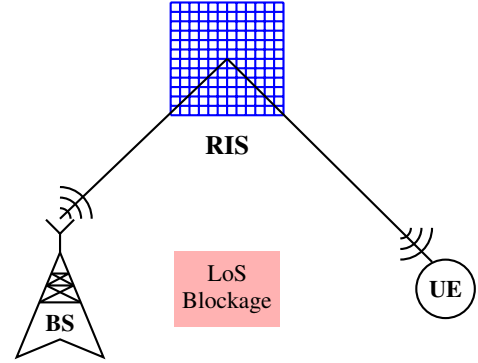
	Critical to the use of \ac{RIS} is the control of the \ac{RIS} elements, commonly through phase shifters. These phase shifters provide element-by-element control with a certain resolution, and allow the designer to coherently add up paths to or from the \ac{RIS} \cite{RIS_sumrate_2020,RIS_commag_2021}. For localization, in contrast to communication, the receiver should know the values of the \ac{RIS} phase profiles to apply suitable high-accuracy processing methods \cite{wymeersch2020radio}. The ability to modulate the \ac{RIS} phase profiles has additional benefits, such as separating the controlled and uncontrolled multipath through temporal coding \cite{keykhosravi2021multiris}. Hence, the ability to control in a precise and known manner is essential for \ac{ISLAC} applications, which necessitates the availability of accurate and simple \ac{RIS} phase control models. Such models should ideally account for the per-element response \cite{abeywickrama2020intelligent}, the finite quantization of the control \cite{RIS_phase_quantization_2021,RIS_commag_2021}, mutual coupling \cite{RIS_coupling_2021}, calibration effects, as well as power losses. Most studies on \ac{RIS} localization have considered ideal phase shifters (e.g. \cite{Keykhosravi2020_SisoRIS,RIS_bounds_TSP_2021,RIS_loc_2021_TWC,Shaban2021,rahal2021risenabled}), ignoring the listed impairments. How these proposed localization studies fare under these impairments is both unknown and important. 
	
	In this paper, we study the impact of realistic \ac{RIS} control models \cite{abeywickrama2020intelligent} in a geometric near-field scenario with \ac{LoS} blockage and a single-antenna \ac{UE} and \ac{BS}. In particular, we employ the \ac{MCRB} \cite{Fortunati2017} as a tool to assess the loss of localization performance under model mismatch, where the algorithm assumes an ideal phase control model, different from the real phase control \cite{abeywickrama2020intelligent}. Our contributions are as follows: \emph{(i)} we provide a simple expression to find the pseudo-true parameter for the considered scenario; \emph{(ii)} we derive the \ac{MCRB} of the pseudo-true parameter and the \ac{LB} of the true parameter; \emph{(iii)} we evaluate the \ac{MCRB} and \ac{LB} and compare with the \ac{MML} estimator, showing that at high \ac{SNR}, the model mismatch can lead to an order-of-magnitude localization performance degradation, both in terms of the \ac{LB} and the matching \ac{MML}. In contrast, when the true phase control model is available, localization performance is relatively stable, for all considered model parameter settings.

	\section{System Model}\label{sec:System}
	In this section, we describe the system geometry, present the signal model and the RIS phase shift model, and formulate the problem of interest.

	\subsection{Geometry and Signal Model}
	We consider a RIS-aided localization system (see Fig.~\ref{fig:RISconfiguration}) with a single-antenna BS, an $M$-element RIS, and a single-antenna UE having the following three-dimensional locations: $\bp_{\text{BS}}$ denotes the known BS location, $\bp_{\text{RIS}}= [\rmx_{\text{R}}\, \rmy_{\text{R}}\, \rmz_{\text{R}}]^{\intercal}$ is the known center of the RIS, $\bp_{m} = [\rmx_m\, \rmy_m \, \rmz_m]^{\intercal}$ represents the known location of the $m$th RIS element for $1\leq m\leq M$, and $\bbp = [\rmx_\text{UE}\, \rmy_\text{UE}\, \rmz_\text{UE}]^{\intercal}$ is the unknown UE location.  
	
	Without loss of generality, the BS broadcasts a narrowband signal $s_t$ over $T$ transmissions under the constraint of $\mathbb{E}\{\abs{s_t}^2\} = E_s$. Assuming \ac{LoS} blockage and the absence of uncontrolled multipath, the signal received by the UE involves only reflections from the RIS and can be expressed at transmission $t$ as 
	\begin{equation}\label{eq_yt}
		y_t  = \bar{\alpha}\underbrace{\ba^{\intercal}(\bbp) \text{diag}(\bw_t) \ba(\bp_{\text{BS}})}_{\triangleq \bb^{\intercal}(\bbp)\bw_t}s_t + n_t\,, 
	\end{equation}
	where $\bar{\alpha}$ is the unknown channel gain, $\bw_t = [w_{t,1}\, \ldots\, w_{t,M}]^{\intercal}$ is the RIS phase profile at transmission $t$, and $n_t$ is uncorrelated zero-mean additive Gaussian noise with variance $N_0/2$ per real dimension. 
	Moreover, $\bb(\bbp) = \ba(\bbp) \odot \ba(\bp_{\text{BS}})$, where $\ba(\bp)$ is the near-field RIS steering vector for a given position $\bp$, defined as
	$[\ba(\bp)]_{m} = \exp(-j {2\pi}(\norm{\bp-\bp_m}-\norm{\bp-\bp_{\text{RIS}}})/{\lambda})$, 
	for $m\in\{1, \ldots, M\}$, in which $\lambda$ denotes the signal wavelength. 
	
	\subsection{RIS Phase Shift Model}
	Following the practical phase shift model in \cite{abeywickrama2020intelligent}, we consider \textit{phase-dependent amplitude variations} of the RIS elements, given by 
	\begin{align}\label{eq_wtm_ris}
		w_{t,m} = \beta(\theta_{t,m}) e^{j\theta_{t,m}},
	\end{align}
	with $\theta_{t,m}\in [-\pi, \pi)$ and $\beta(\theta_{t,m})\in[0,1]$ denoting the phase shift and the corresponding amplitude, respectively. In particular, $\beta(\theta_{t,m})$ is expressed as
	\begin{equation}\label{eq_beta_model}
		\beta(\theta_{t,m}) = (1-\beta_{\text{min}})\left(\frac{\sin(\theta_{t,m}-\phi) + 1}{2}\right)^{\kappa} + \beta_{\text{min}},
	\end{equation}
	where $\beta_{\text{min}}\geq 0$, $\phi\geq 0$, and $\kappa\geq 0$ are the constants related to the specific circuit implementation \cite{abeywickrama2020intelligent}. 
	
	\subsection{Problem Description}
	Given the observations in $y_t$ from \eqref{eq_yt} over $T$ transmission instances, our goal  is to derive theoretical performance bounds for estimating the position of the UE under mismatch between the \textit{true model} in \eqref{eq_wtm_ris} and \textit{the assumed model} with $\tilde{w}_{t,m}= \exp({j\theta_{t,m}})$ (which is equivalent to assuming $\beta_{\text{min}} = 1$). 
	In other words, we aim to quantify localization performance loss due to this model mismatch resulting from the RIS hardware impairment specified in \eqref{eq_wtm_ris}. To that end, we will resort to the MCRB analysis \cite{Fortunati2017, Fortunati2018Chapter4,MCRB_TSP_2021,MCRB_delay_ICASSP_2020}, as discussed in the next section.

	\section{Misspecified Cram\'{e}r-Rao Bound Analysis}
	In this section, we introduce several notations and definitions, including the \ac{MCRB} and \ac{LB}, before deriving the \ac{MCRB} for our scenario. 
	\subsection{Preliminaries}
	
	The likelihood functions under the true and assumed models are given, respectively, by
	\begin{align}\label{eq:truedpf}
		p(\by| \bet) & = \frac{1}{(\pi N_0)^T} e^{-\frac{1}{N_0} \norm{\by-\bmu(\bet)}^2},\\
		\tilde{p}(\by|\bet)  & = \frac{1}{(\pi N_0)^T} e^{-\frac{1}{N_0} \norm{\by-\tilde{\bmu}(\bet)}^2},\label{eq:assumedpdf}
	\end{align}
	where $\by \triangleq [y_1\, \ldots \,y_T]^{\intercal}$, $\bmu(\bet) \triangleq [\mu_1(\bet)\, \ldots \, \mu_T(\bet)]^{\intercal}$, and $\tbmu(\bet) \triangleq [\tmu_1(\bet)\, \ldots \, \tmu_T(\bet)]^{\intercal}$. Also, the noise-free observations under the true and assumed models are
	\begin{align}
		\mu_t(\bet) &  \triangleq {\alpha} \sum_{m=1}^{M} [\bb(\bp)]_m w_{t,m} s_t,  \label{eq:mut}\\
		\tmu_t(\bet) & \triangleq  {\alpha} \sum_{m=1}^{M} [\bb(\bp)]_m \tilde{w}_{t,m} s_t,  \label{eq:mut_til}
	\end{align}
	where $w_{t,m}$ is defined in \eqref{eq_wtm_ris}, while under the assumed model $\tilde{w}_{t,m}= \exp({j\theta_{t,m}})$. 
	
	\subsection{MCRB Definition}
	We first introduce the pseudo-true parameter  
	\begin{equation}\label{eq:eta0}
		\bet_0 = \argmin_{\bet\in\mathbb{R}^{5}} D\left(p(\by|\bbet)\Vert \tilde{p}(\by| \bet)\right),
	\end{equation}
	where $ D\left(p(\by|\bbet)\Vert \tilde{p}(\by|\bet)\right)$ denotes the Kullback-Leibler (KL) divergence between the distributions $p(\by|\bbet)$ and $\tilde{p}(\by|\bet)$. 
	Next, let  $\hat{\bet}(\by)$ be a misspecified-unbiased (MS-unbiased) estimator of $\bbet$, i.e., the mean of the estimator $\hat{\bet}(\by)$ under the true model is equal to $\bet_0$. Then, the \ac{MCRB} is a lower bound for the covariance matrix of any MS-unbiased estimator of $\bbet$, $\hat{\bet}(\by)$ \cite{Fortunati2017, Fortunati2018Chapter4, Ricmond2015MCRB}:
	\begin{align}
		&\mathbb{E}_p\{(\hat{\bet}(\by)-\bet_0)(\hat{\bet}(\by)-\bet_0)^{\intercal}\} \succeq  \text{MCRB}(\bet_0),
		\label{eq:MCRB}
	\end{align}
	where $\mathbb{E}_p\{\cdot\}$ denotes the expectation operator under the true model $p(\by|\bbet)$ and 
	\begin{align}
		\text{MCRB}(\bet_0) \triangleq \bA_{\bet_0}^{-1} \bB_{\bet_0}\bA_{\bet_0}^{-1},
	\end{align}
	in which the $(i,j)$-th elements of the matrices $\bA_{\bet_0}$ and  $\bB_{\bet_0}$ are calculated as 
	\begin{align}\label{eq:Aeta0}
		[\bA_{\bet_0}]_{ij} &= \mathbb{E}_p\left\{\frac{\partial^2}{\partial \eta_i \partial \eta_j} \ln \tilde{p}(\by|\bet)\Big|_{\bet = \bet_0}\right\}, \\\label{eq:Beta0}
		[\bB_{\bet_0}]_{ij} &= \mathbb{E}_p\left\{\frac{\partial}{\partial \eta_i } \ln \tilde{p}(\by|\bet) \frac{\partial}{\partial \eta_j } \ln \tilde{p}(\by|\bet)\Big|_{\bet =  \bet_0}\right\},
	\end{align} 
	for $1\leq i, j\leq 5$, with $\eta_i$ denoting the $i$th element of $\bet$. Note that without model mismatch, 
	$\bA_{\bet_0} = \bA_{\bbet} = -\bB_{\bet_0} = -\bB_{\bbet}$ so that the \ac{MCRB} reverts to the classical Cram\'{e}r-Rao bound (CRB) \cite{Fortunati2018Chapter4}. 
	
	Since the value of the pseudo-true parameter is generally not of interest, the \ac{MCRB} is used to establish the \ac{LB} of any MS-unbiased estimator with respect to the true parameter value \cite{Fortunati2017}
	\begin{align}
		\mathbb{E}_p\{(\hat{\bet}(\by)-\bbet)(\hat{\bet}(\by)-\bbet)^{\intercal}\} \succeq \text{LB}(\bet_0),\label{eq:LB}
	\end{align}
	where $\text{LB}(\bet_0)\triangleq \text{MCRB}(\bet_0)  + (\bbet-\bet_0)(\bbet-\bet_0)^{\intercal}$. 
	The last term is a bias term; that is, $\text{Bias}(\bet_0) \triangleq (\bbet-\bet_0)(\bbet-\bet_0)^{\intercal}$, and it is independent of the SNR. Hence, as the SNR tends to infinity, the MCRB term goes to zero, and the bias term becomes a tight bound for the MSE of any MS-unbiased estimator.

	\subsection{MCRB Derivation for RIS-aided Localization}\label{sec:MCRBder}

	
	\subsubsection{Determining the Pseudo-True Parameter}
	To derive the MCRB for estimating the UE position under mismatch between the amplitude models for the RIS elements, we should first calculate the $\bet_0$ parameter in \eqref{eq:eta0} for the system model described in Section~\ref{sec:System}; that is, we should find the value of $\bet$ that minimizes the KL divergence between $p(\by|\bbet)$ in \eqref{eq:truedpf} and $\tilde{p}(\by|\bet)$ in \eqref{eq:assumedpdf}. The following lemma characterizes $\bet_0$ for the considered system model.
	\begin{lemma}
		The value of $\bet$ that minimizes the KL divergence between $p(\by|\bbet)$ in \eqref{eq:truedpf} and $\tilde{p}(\by|\bet)$, which is a parameterized version of \eqref{eq:assumedpdf}, can be expressed as
		\begin{equation} \label{eq_betz}
			\bet_0 = \argmin_{\bet\in\mathbb{R}^5} \norm{\bmu(\bbet)-\tbmu(\bet)}.
		\end{equation}
	\end{lemma}
	\begin{proof}
		See Appendix~\ref{sec:AppA}.
	\end{proof}
	
	Hence, the pseudo-true parameter minimizes the Euclidean distance  between the noise-free observation under the true and assumed models, with respect to the assumed model.

	Let $\gamma(\bet)\triangleq  \norm{\bmu(\bbet)-\tbmu(\bet)}$. It is noted from \eqref{eq:mut} and \eqref{eq:mut_til} that $\gamma(\bet)$ is non-convex with respect to $\bet$; hence, it is challenging to solve \eqref{eq_betz} in its current form. Based on \eqref{eq:assumedpdf} and \eqref{eq:mut_til}, we can re-write \eqref{eq_betz} as
	\begin{align}\label{eq_etabar2}
		(\balp_0, \bp_0) = \argmin_{(\balp, \bp)} \norm{\bmu(\bbet)- \alpha \, \bcc(\bp) } ,
	\end{align}
	where $[\bcc(\bp)]_t \triangleq \sum_{m=1}^{M} [\bb(\bp)]_m \tilde{w}_{t,m} s_t$.
	The optimal complex-valued $\alpha$ for any given $\bp$ can be expressed 
	as 
	\begin{gather}\label{eq:alpOpt}
		\alpha = \frac{ \bcc^{\mathsf{H}}(\bp) \bmu(\bbet) }{ \bcc^{\mathsf{H}}(\bp) \bcc(\bp)  } \,\cdot
	\end{gather}
	Inserting \eqref{eq:alpOpt} into \eqref{eq_etabar2}, the minimization problem can be reduced to a three-dimensional search as follows:
	\begin{gather}\label{eq_etabar3}
		\bp_0 = \argmax_{\bp} \bmu(\bbet)^{\mathsf{H}} \projrange{\bcc(\bp)} \bmu(\bbet) ~,
	\end{gather}
	where $\projrange{\boldsymbol{\rmx}} \triangleq {\boldsymbol{\rmx} \boldsymbol{\rmx}^{\mathsf{H}}}/{\boldsymbol{\rmx}^{\mathsf{H}} \boldsymbol{\rmx}}$. Therefore, $\bet_0=[\balp_0^{\intercal}~\bp_0^{\intercal}]^{\intercal}$ can be found by first performing a three-dimensional optimization as in \eqref{eq_etabar3}, and then calculating $\alpha_0$ via \eqref{eq:alpOpt} and obtaining $\balp_0$ as $\balp_0 = [\text{Re}(\alpha_0)~ \text{Im}(\alpha_0)]^{\intercal}$.
	
	\begin{rem}
		Note that to determine the pseudo-true parameter, the true parameter (including the location $\bbp$) is known; hence, it can be used to initialize the optimization problem \eqref{eq_etabar3}, significantly reducing the computational complexity. 
	\end{rem}

	\subsubsection{Deriving the MCRB}
	After finding $\bet_0$, we compute the matrices  $\bA_{\bet_0}$ from \eqref{eq:Aeta0} and  $\bB_{\bet_0}$ from \eqref{eq:Beta0} for evaluating the MCRB in \eqref{eq:MCRB}. Due to page limitation, details of computation of $\bA_{\bet_0}$ and  $\bB_{\bet_0}$ are not presented in this manuscript.	Based on the pdf expressions in \eqref{eq:truedpf}--\eqref{eq:assumedpdf}, \eqref{eq:Aeta0} becomes, with $\bep(\bbet)=\bmu(\bbet)-\tbmu(\bet)$
	\begin{align}\label{eq:Aeta0_1}
		& 	[\bA_{\bet_0}]_{ij}  = \int p(\by|\bbet)\frac{\partial^2}{\partial \eta_i \partial \eta_j} \ln \tilde{p}(\by|\bet)\Big|_{\bet = \bet_0}\, \text{d}\by \\
		& = \frac{2}{N_0} \Re\left\{\bep(\bbet)^{\mathsf{H}} \frac{\partial^2\tbmu(\bet)}{\partial\eta_i \partial\eta_j}-\Big(\frac{\partial\tbmu(\bet)}{\partial\eta_i}\Big)^{\mathsf{H}}\frac{\partial\tbmu(\bet)}{\partial\eta_j}\right\}\Big|_{\bet = \bet_0},
	\end{align}
	where  $\frac{\partial^2\tbmu(\bet)}{\partial\eta_i \partial\eta_j}\triangleq \left[\frac{\partial^2 \tmu_1(\bet)}{\partial \eta_i \partial \eta_j} \, \ldots \, \frac{\partial^2 \tmu_T(\bet)}{\partial \eta_i \partial \eta_j}\right]^{\intercal}$. 
	%
	Similarly, after some algebraic manipulation, the $(i,j)$th entry of matrix $\bB_{\bet_0}$ in \eqref{eq:Beta0} can be written as
	\begin{align}
		[\bB_{\bet_0}]_{ij} &= \frac{2}{N_0}\Bigg[ \frac{2}{N_0}\Re\left\{\bep(\bbet)^{\mathsf{H}}\frac{\partial\tbmu(\bbet)}{\partial\eta_i}\right\} \Re\left\{\bep(\bbet)^{\mathsf{H}}\frac{\partial\tbmu(\bbet)}{\partial\eta_j}\right\}  \notag \\
		&+ \Re \left\{\left(\frac{\partial\tbmu(\bbet)}{\partial\eta_i}\right)^{\mathsf{H}}\frac{\partial\tbmu(\bbet)}{\partial\eta_j}\right\}\Bigg]\Bigg|_{\bet = \bet_0}.
	\end{align}
	Therefore, once we have  computed the first and the second derivatives of $\tmu_t(\bet)$ with respect to $\bet$, we can easily compute the matrices $\bA_{\bet_0}$ and $\bB_{\bet_0}$ as specified above. The derivatives are presented in Appendix \ref{sec:AppB}. Based on $\bA_{\bet_0}$ and $\bB_{\bet_0}$, the MCRB in \eqref{eq:MCRB} and the lower bound in \eqref{eq:LB} can be evaluated in a straightforward manner. 

	\section{Mismatched Estimator}
	\subsection{Definition and Relation to MCRB}
	We introduce the \ac{MML} estimator as \cite{Fortunati2017}
	\begin{equation}\label{eq:MML1}
		\hat{\bet}_{\text{MML}}(\by) = \argmax_{\bet\in\mathbb{R}^5} \ln \tilde{p}(\by|\bet). 
	\end{equation}
	Under some regularity conditions, it is shown that $\hat{\bet}_{\text{MML}}(\by)$ is asymptotically MS-unbiased and its error covariance matrix is asymptotically equal to the $\text{MCRB}(\bet_0)$ \cite[Thm. 2]{Fortunati2017}. Hence, the covariance matrix of the MML estimator is asymptotically tight to the MCRB. 
	
	\subsection{MML Estimator for UE Location}
	We now investigate the MML estimator for the UE position under model misspecification regarding the RIS amplitudes. From \eqref{eq:assumedpdf} and \eqref{eq:MML1}, the MML estimator based on the received signal $\by$ can be expressed as $\hat{\bet}_{\text{MML}}(\by)=\argmin_{\bet\in\mathbb{R}^5} \norm{\by-\tbmu(\bet)}$. 
	Since this problem is in the same form as the optimization problem in \eqref{eq_betz}, it can be reduced to a three-dimensional optimization problem as discussed in Section~\ref{sec:MCRBder}. In order to solve the resulting problem, initialization can be very critical as we are facing with a non-convex optimization problem. During the estimation process, we do not have access to the true position $\bbp$. Hence, we cannot use the true position vector $\bbp$ for the initialization. If an arbitrarily chosen position vector is used for the initialization, the global optimal solution of \eqref{eq:MML1} is not always obtained. To find a remedy for this issue, we use the Jacobi-Anger expansion approach discussed in \cite{Shaban2021} to obtain an initial position vector rather than starting from an arbitrarily generated position vector.  
	In particular, for the position vector $\bbp = \norm{\bbp}[\sin(\vartheta)\cos(\varphi) \, \sin(\vartheta)\sin(\varphi)\, \cos(\vartheta) ]^{\intercal}$, for some $N\in\mathbb{N}$,  $\ba(\bbp)$ is approximated as 
	\begin{equation}\label{eq:JacobiAnger}
	    [\ba(\bbp)]_m \approx \sum_{n=-N}^{N} j^{n} J_n\left(-\frac{2\pi}{\lambda} \norm{\bp_m} \sin(\vartheta)\right) e^{jn(\varphi-\psi_m)},
	\end{equation}
	where $\bp_m = \norm{\bp_m} [\cos(\psi_m)\, \sin(\psi_m) \, 0]^{\intercal}$ and $J_n(\cdot)$ is the $n$th order Bessel function of the first kind. After defining $\boldsymbol{G}(\varphi)$ and $\boldsymbol{h}(\varphi)$  exactly as in \cite{Shaban2021}, we can rewrite $y_t$ as
	\begin{equation}
	    y_t \approx \bar{\alpha} \boldsymbol{h}^{\intercal}(\varphi) \boldsymbol{G}(\vartheta) \text{diag}(\bw_t) \ba(\bp_{\text{BS}})s_t + n_t,
	\end{equation}
	By using two-step simple line searches given in \cite[Eqs.~ 31, 32]{Shaban2021}, estimates of $\varphi$ and $\vartheta$ are obtained. Let $\hat{\vartheta}$ and $\hat{\varphi}$ denote the estimates of $\vartheta$ and $\varphi$.  After these two steps, we generate random variables $\tilde{d}$ such that 
	$\tilde{d}[\sin(\hat{\vartheta})\cos(\hat{\varphi}) \, \sin(\hat{\vartheta})\sin(\hat{\varphi})\, \cos(\hat{\vartheta}) ]^{\intercal}$ are initial position vectors for any $\tilde{d}$. 
	As discussed in \cite{Shaban2021}, this two-step simple line searches have a low computational complexity and do not add any complexity cost to the MML estimator.

	\section{Numerical Results}
	
	In this section, we first present numerical examples for evaluating the lower bounds in various scenarios, and then compare the performance of the MML estimator against the lower bounds. 
	
	
	\subsection{Simulation Setup}
	We consider an RIS with $M = 2500$ elements, where the inter-element spacing is $\lambda/2$ and the area of each element is $A = \lambda^2/4$ \cite{Shaban2021}. The carrier frequency is equal to $f_c = 28\,$ GHz. The RIS is modeled to lie in the X-Y plane with $\bp_{\text{RIS}}=[0\, 0\, 0]^{\intercal}$. Moreover, for the phase profile, the $\theta_{t,m}$ values are generated uniformly and independently between $-\pi$ and $\pi$. For the  model of the RIS in \eqref{eq_beta_model}, $\kappa=2$ and $\phi=0$. The BS is located at $\bp_{\text{BS}}=5.77\times 
	[-1\, 1\, 1]^{\intercal}$ meters.  For given distance $d$ to the RIS, the UE is located at $d\times {[1\, 1\, 1]^{\intercal}}/{\norm{[1\, 1\, 1]}}$ meters. We set the number of transmission to $T = 50$. For simplicity, we assume that $s_t = \sqrt{E_s}$ for any $t$. The  SNR is defined as
	\begin{gather}
		\text{SNR} = \frac{E_s \norm{\bbalp}^2}{T N_0} \sum_{t=1}^{T}  \abs{\bb^{\intercal}(\bbp) \bw_t}^2\,.
	\end{gather}
	
	\subsection{Analysis of Lower Bounds and MML Estimator}

	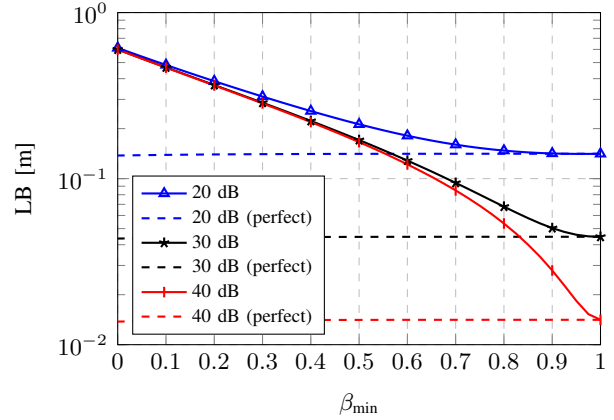
\begin{figure}
		\centering
		\begin{tikzpicture}
			\begin{semilogyaxis}[
				width=8cm,
            height=6cm,
				legend style={nodes={scale= 0.8, transform shape},at={(1,1)},anchor=north east}, 
				legend cell align={left},
				legend image post style={mark indices={}},
				xlabel={$\beta_{\text{min}}$},
				ylabel={LB [m]},
				xmin=0, xmax=1,
				ymin=0.01, ymax=1,
				xtick={0, 0.1, 0.2, 0.3, 0.4, 0.5, 0.6, 0.7, 0.8, 0.9, 1},
				ytick={0.01,0.1,1},
				legend pos=south west,
				ymajorgrids=true,
				xmajorgrids=true,
				grid style=dashed,
				]
				
				\addplot[thick,
				color=blue,
				mark = triangle,
				mark indices={1,5,9,13,17,21,25,29,33,37,41},
				mark options={solid},
				]
				coordinates {
					(1,0.14117)(0.975,0.1409)(0.95,0.1409)(0.925,0.14117)(0.9,0.14174)(0.875,0.14263)(0.85,0.14387)(0.825,0.14549)(0.8,0.14752)(0.775,0.14997)(0.75,0.15288)(0.725,0.15628)(0.7,0.16019)(0.675,0.16462)(0.65,0.16958)(0.625,0.17513)(0.6,0.18131)(0.575,0.18807)(0.55,0.19555)(0.525,0.20362)(0.5,0.21243)(0.475,0.22197)(0.45,0.2323)(0.425,0.24335)(0.4,0.25526)(0.375,0.26807)(0.35,0.28176)(0.325,0.29641)(0.3,0.31206)(0.275,0.32883)(0.25,0.34673)(0.225,0.36589)(0.2,0.38636)(0.175,0.40824)(0.15,0.43166)(0.125,0.45675)(0.1,0.4836)(0.075,0.51242)(0.05,0.54336)(0.025,0.57662)(0,0.61241)
				};
				
				\addplot[thick, dashed,
				color=blue,
				mark options={solid},
				]
				coordinates {
					(1,0.14117)(0.975,0.14117)(0.95,0.14118)(0.925,0.14118)(0.9,0.14119)(0.875,0.14119)(0.85,0.14119)(0.825,0.14119)(0.8,0.14119)(0.775,0.14119)(0.75,0.14118)(0.725,0.14117)(0.7,0.14116)(0.675,0.14115)(0.65,0.14114)(0.625,0.14112)(0.6,0.1411)(0.575,0.14107)(0.55,0.14104)(0.525,0.14101)(0.5,0.14097)(0.475,0.14093)(0.45,0.14088)(0.425,0.14082)(0.4,0.14076)(0.375,0.14068)(0.35,0.1406)(0.325,0.14051)(0.3,0.1404)(0.275,0.14028)(0.25,0.14015)(0.225,0.14)(0.2,0.13984)(0.175,0.13966)(0.15,0.13945)(0.125,0.13923)(0.1,0.13898)(0.075,0.13871)(0.05,0.13841)(0.025,0.13809)(0,0.13774)
				};

				\addplot[thick,
				color=black,
				mark = star,
				mark indices={1,5,9,13,17,21,25,29,33,37,41},
				mark options={solid},
				]
				coordinates {
					(1,0.044641)(0.975,0.044906)(0.95,0.045923)(0.925,0.047774)(0.9,0.050376)(0.875,0.053731)(0.85,0.0578)(0.825,0.062499)(0.8,0.067806)(0.775,0.073668)(0.75,0.08005)(0.725,0.086946)(0.7,0.094314)(0.675,0.10209)(0.65,0.11046)(0.625,0.11916)(0.6,0.12846)(0.575,0.13818)(0.55,0.14848)(0.525,0.15932)(0.5,0.17076)(0.475,0.18271)(0.45,0.19531)(0.425,0.20858)(0.4,0.22254)(0.375,0.23726)(0.35,0.2528)(0.325,0.26916)(0.3,0.28643)(0.275,0.30471)(0.25,0.32402)(0.225,0.34449)(0.2,0.36621)(0.175,0.38925)(0.15,0.41377)(0.125,0.43984)(0.1,0.46766)(0.075,0.49736)(0.05,0.52911)(0.025,0.56312)(0,0.59962)
				};
				
				\addplot[thick,
				color=black,
				style = dashed,
				mark options={solid},
				]
				coordinates {
					(1,0.044641)(0.975,0.044643)(0.95,0.044645)(0.925,0.044646)(0.9,0.044647)(0.875,0.044648)(0.85,0.044649)(0.825,0.044648)(0.8,0.044648)(0.775,0.044647)(0.75,0.044645)(0.725,0.044643)(0.7,0.04464)(0.675,0.044636)(0.65,0.044631)(0.625,0.044626)(0.6,0.044619)(0.575,0.044611)(0.55,0.044602)(0.525,0.044591)(0.5,0.044579)(0.475,0.044565)(0.45,0.044549)(0.425,0.044531)(0.4,0.044511)(0.375,0.044488)(0.35,0.044462)(0.325,0.044432)(0.3,0.044399)(0.275,0.044362)(0.25,0.04432)(0.225,0.044273)(0.2,0.044221)(0.175,0.044163)(0.15,0.044098)(0.125,0.044027)(0.1,0.043949)(0.075,0.043863)(0.05,0.043769)(0.025,0.043667)(0,0.043557)
				};

				\addplot[	thick,
				color=red,
				mark = |,
				mark options={solid},
				mark indices={1,5,9,13,17,21,25,29,33,37,41}
				]
				coordinates {
					(1,0.014117)(0.975,0.015269)(0.95,0.01836)(0.925,0.022755)(0.9,0.027994)(0.875,0.033804)(0.85,0.040002)(0.825,0.046691)(0.8,0.053679)(0.775,0.061019)(0.75,0.068665)(0.725,0.076626)(0.7,0.084953)(0.675,0.093571)(0.65,0.10265)(0.625,0.11199)(0.6,0.12188)(0.575,0.13215)(0.55,0.14301)(0.525,0.15421)(0.5,0.16595)(0.475,0.17829)(0.45,0.19124)(0.425,0.20476)(0.4,0.21901)(0.375,0.23398)(0.35,0.24969)(0.325,0.26628)(0.3,0.28374)(0.275,0.30217)(0.25,0.32167)(0.225,0.34229)(0.2,0.36412)(0.175,0.3873)(0.15,0.41193)(0.125,0.43813)(0.1,0.46605)(0.075,0.49583)(0.05,0.52768)(0.025,0.56178)(0,0.59833)
				};
				
				\addplot[thick,
				color=red,
				style= dashed,
				mark options={solid},
				]
				coordinates {
					(1,0.014117)(0.975,0.014117)(0.95,0.014118)(0.925,0.014118)(0.9,0.014119)(0.875,0.014119)(0.85,0.014119)(0.825,0.014119)(0.8,0.014119)(0.775,0.014119)(0.75,0.014118)(0.725,0.014117)(0.7,0.014116)(0.675,0.014115)(0.65,0.014114)(0.625,0.014112)(0.6,0.01411)(0.575,0.014107)(0.55,0.014104)(0.525,0.014101)(0.5,0.014097)(0.475,0.014093)(0.45,0.014088)(0.425,0.014082)(0.4,0.014076)(0.375,0.014068)(0.35,0.01406)(0.325,0.014051)(0.3,0.01404)(0.275,0.014028)(0.25,0.014015)(0.225,0.014)(0.2,0.013984)(0.175,0.013966)(0.15,0.013945)(0.125,0.013923)(0.1,0.013898)(0.075,0.013871)(0.05,0.013841)(0.025,0.013809)(0,0.013774)
				};

				\legend{20 dB, 20 dB (perfect), 30 dB, 30 dB (perfect), 40 dB, 40 dB (perfect) }
				
			\end{semilogyaxis}
		\end{tikzpicture}
		\caption{LB versus $\beta_{\text{min}}$ for SNR = $20\,$dB, $30\,$dB and $40\,$dB when the UE distance is $5$ meters. The curves marked (perfect) are nearly flat and correspond to the the assumed model being equal to the true model for different $\beta_{\text{min}}$.}
		\label{fig:1}
	\end{figure}

	
	To solve \eqref{eq_etabar3} for LB computation, we employ the GlobalSearch algorithm of MATLAB by providing $10$ different initial points. These initial points are generated by multiplying each component of the true position of the UE, $\bbp$, by independent uniform random variables between zero and one.
	
		For the initialization of the MML estimator, $N$ in \eqref{eq:JacobiAnger} is taken as $5$ for using the Jacobi-Anger expansion approach. Let $\hat{\vartheta}$ and $\hat{\varphi}$ denote the estimates of $\vartheta$ and $\varphi$.  After obtaining the the estimates of $\vartheta$ and $\varphi$, for the norm of the $\bbp$, we generate a random variable $\tilde{d}$, between $0$ and $1000$ for 10 different seeds, and we run the GlobalSearch algorithm of MATLAB for each initial point.

	\subsection{Results and Discussion}
	
	In Fig.~\ref{fig:1}, we show $\text{LB}$ as a function of $\beta_{\text{min}}$ for SNRs of $20$, $30$, and $40$ dB when the UE distance is $5$ meters from the center of the RIS. In addition, for comparison purposes, the lower bounds in the presence of the perfect knowledge of the $\beta_{\text{min}}$ values are also presented (marked as ``Perfect" in the figure).\footnote{To obtain these curves, it is assumed that for any given $\beta_{\text{min}}$, the perfect knowledge of the RIS amplitudes is available. Hence, this bound reduces to the classical CRB expression as in \cite[Eq.~ 9]{Shaban2021}. As $\bmu$ changes with respect to the true value of $\beta_{\text{min}}$, this expression is computed for each $\beta_{\text{min}}$ value.} We observe from the figure that as $\beta_{\text{min}}$ decreases, i.e., as the mismatch between the true and the assumed models increases, the LB increases and raising the SNR level does not improve the LB values significantly. In addition, the sensitivity to the model mismatch is more pronounced at higher SNR, while for an SNR of  $20$ dB, the performance is relatively insensitive for $\beta_{\text{min}}>0.7$. Interestingly, we note that when the true model is known (i.e., in the presence of perfect knowledge of $\beta_{\text{min}}$), the value of $\beta_{\text{min}}$ does not influence the LB values notably.

		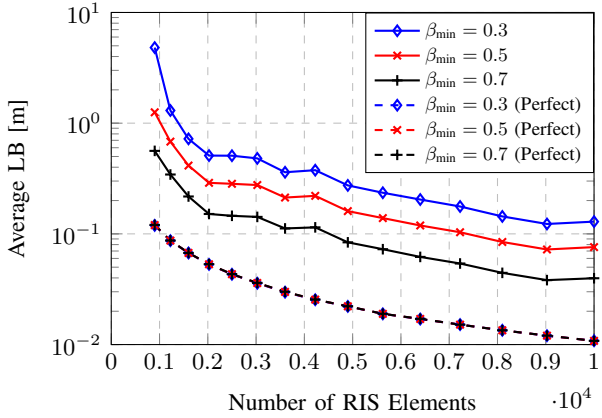
\begin{figure}
		\centering
		\begin{tikzpicture}
			\begin{semilogyaxis}[
			width=8cm,
            height=6cm,
				legend style={nodes={scale= 0.8, transform shape},at={(1,1)},anchor=north east}, 
				legend cell align={left},
				xlabel={Number of RIS Elements},
				ylabel={Average LB [m]},
				xmin=0, xmax=10^4,
				ymin=0.01, ymax=10,
				xtick={0, 1000, 2000, 3000, 4000, 5000, 6000, 7000, 8000, 9000,10000},
				ytick={0.01,0.1,1,10},
				ymajorgrids=true,
				xmajorgrids=true,
				grid style=dashed,
			    every axis plot/.append style={thick},
				]
				
				\addplot[
				color=blue,
				mark = diamond,
				mark options={solid},
				]
				coordinates {
					(900,4.8171)(1225,1.3019)(1600,0.72145)(2025,0.50953)(2500,0.50826)(3025,0.48078)(3600,0.36028)(4225,0.37616)(4900,0.27408)(5625,0.23523)(6400,0.20428)(7225,0.17677)(8100,0.14375)(9025,0.12286)(10000,0.12901)
				};
				
				\addplot[
				color=red,
				mark = x,
				mark options={solid},
				]
				coordinates {
					(900,1.2537)(1225,0.68138)(1600,0.4141)(2025,0.28915)(2500,0.28333)(3025,0.27614)(3600,0.21288)(4225,0.22055)(4900,0.16056)(5625,0.13885)(6400,0.11908)(7225,0.10345)(8100,0.084466)(9025,0.0723)(10000,0.075931)
				};
				
				\addplot[
				color=black,
				mark = +,
				mark options={solid},
				]
				coordinates {
					(900,0.56195)(1225,0.34373)(1600,0.21763)(2025,0.15127)(2500,0.1456)(3025,0.14231)(3600,0.11195)(4225,0.1145)(4900,0.083857)(5625,0.072691)(6400,0.061925)(7225,0.05399)(8100,0.044449)(9025,0.038218)(10000,0.039734)
				};
				
				
				\addplot[
				color=blue,
				mark = diamond,
				style = dashed,
				mark options={solid},
				]
				coordinates {
				(900,0.11927)(1225,0.086862)(1600,0.067211)(2025,0.052954)(2500,0.043148)(3025,0.035793)(3600,0.029751)(4225,0.0254)(4900,0.022124)(5625,0.018948)(6400,0.016996)(7225,0.015154)(8100,0.013463)(9025,0.012028)(10000,0.010822)
				};
				
				\addplot[
				color=red,
				mark = x,
				style = dashed,
				mark options={solid},
				]
				coordinates {
    			(900,0.12003)(1225,0.086994)(1600,0.067277)(2025,0.053129)(2500,0.043244)(3025,0.035978)(3600,0.029913)(4225,0.025462)(4900,0.022169)(5625,0.018948)(6400,0.01704)(7225,0.015137)(8100,0.013459)(9025,0.012038)(10000,0.010802)
				};
				
				\addplot[
				color=black,
				mark = +,
				style = dashed,
				mark options={solid},
				]
				coordinates {
    			(900,0.12045)(1225,0.087063)(1600,0.0673)(2025,0.053252)(2500,0.043297)(3025,0.036081)(3600,0.030011)(4225,0.025507)(4900,0.022194)(5625,0.018965)(6400,0.017077)(7225,0.015127)(8100,0.013455)(9025,0.012049)(10000,0.010801)
				};
				
		\legend{$\beta_{\text{min}} = 0.3$, $\beta_{\text{min}} = 0.5$, $\beta_{\text{min}} = 0.7$, $\beta_{\text{min}} = 0.3$ (Perfect),$\beta_{\text{min}} = 0.5$ (Perfect), $\beta_{\text{min}} = 0.7$ (Perfect)}	
			\end{semilogyaxis}
		\end{tikzpicture}
		\caption{Average LB versus RIS size for $\beta_{\text{min}}\in\{0.3, 0.5, 0.7\}$, when SNR  = 30 dB, and the UE distance is 5 meters.}
		\label{fig:2}
	\end{figure}

	In Fig.~\ref{fig:2}, to understand the impact of the number of RIS elements, by averaging 200 different random phase profiles, average LB values versus RIS size are shown at an SNR of $30\,$dB. We consider $\beta_{\text{min}}\in\{0.3, 0.5, 0.7\}$. Moreover, the average lower bounds under the perfect knowledge of the $\beta_{\text{min}}$ values are plotted. We observe that as the RIS size increases or as $\beta_{\text{min}}$ increases, we obtain lower LB values in general. Interestingly, the curves for different $\beta_{\text{min}}$ values are almost parallel. We note the significant price we pay under mismatch: with perfect knowledge of a RIS with 1000 elements attains similar performance as a RIS with 7000 elements when $\beta_{\text{min}}=0.5$ under mismatch.



	Finally, in Fig.~\ref{fig:3}, we show the performance of the MML estimator versus SNR for $\beta_{\text{min}} = 0.3$ and $0.7$ when the UE distance is $5$ meters. In addition to the performance of the MML estimator, the LB, and the MCRB, the bias term values are also plotted. We observe that the MML estimator exhibits three distinct regimes: a low-SNR regime where MML is limited by noise peaks and thus far away from the LB; a medium-SNR regime where MML is close to the LB, which itself is dominated by the MCRB; and a high-SNR regime, where the MML and LB are limited by the bias term $\text{Bias}(\bet_0)$. 

	\begin{figure}
		\centering
		\begin{tikzpicture}
			\begin{semilogyaxis}[	width=8cm,
            height=6cm,
				legend style={nodes={scale= 0.8, transform shape},at={(1,1)},anchor=north east}, 
				legend cell align={left},
				xlabel={SNR (dB)},
				ylabel={RMSE},
				xmin=-10, xmax=40,
				ymin=0.01, ymax=200,
				xtick={-10,-5,0,5,10,15,20,25,30,35,40},
				ytick={0.01,0.1,1,10,100},
				ymajorgrids=true,
				xmajorgrids=true,
				grid style=dashed,
			    every axis plot/.append style={thick},
				]
				
				\addplot[
				color=blue,
				mark = o, mark options={solid},
				style = dashed
				]
				coordinates {
					(-10,4.3174)(-7.5,3.2381)(-5,2.4288)(-2.5,1.8222)(0,1.3675)(5,0.77215)(10,0.43968)(15,0.25678)(20,0.16018)(25,0.11365)(30,0.094296)(35,0.087311)(40,0.084846)
				};
				
				\addplot[
				color=blue,
				mark = x,mark options={solid},
				style = dashed
				]
				coordinates {
					(-10,77.8379)(-7.5,74.815)(-5,44.9681)(-2.5,21.1869)(0,1.7915)(5,0.78774)(10,0.42852)(15,0.2516)(20,0.16007)(25,0.1155)(30,0.096159)(35,0.088571)(40,0.08569)
				};

				\addplot[
				color=blue,
				mark = square,mark options={solid},
				style = dashed
				]
				coordinates {
					(-10,4.3166)(-7.5,3.237)(-5,2.4274)(-2.5,1.8203)(0,1.365)(5,0.7676)(10,0.43161)(15,0.24272)(20,0.13649)(25,0.076754)(30,0.043161)(35,0.024271)(40,0.01365)
				};

				\addplot[
				color=blue,
				style = dashed
				]
				coordinates {
					(-10,0.083696)(-7.5,0.083693)(-5,0.083719)(-2.5,0.083724)(0,0.083758)(5,0.08372)(10,0.083847)(15,0.083781)(20,0.083829)(25,0.083812)(30,0.083838)(35,0.083869)(40,0.08374)
				};

				\addplot[
                color=red,
				mark = o
				]
				coordinates {
					(-10,4.1387)(-7.5,3.1092)(-5,2.3391)(-2.5,1.7641)(0,1.3361)(5,0.78706)(10,0.50082)(15,0.36641)(20,0.31207)(25,0.29279)(30,0.28647)(35,0.28437)(40,0.28372)
				};

				\addplot[
				color=red,
				mark = x
				]
				coordinates {
					(-10,101.0793)(-7.5,85.1294)(-5,74.3483)(-2.5,29.4231)(0,1.693)(5,0.77907)(10,0.48674)(15,0.36503)(20,0.31541)(25,0.29622)(30,0.28884)(35,0.28589)(40,0.28463)
				};
				
				\addplot[
				color=red,
				mark = square
				]
				coordinates {
					(-10,4.129)(-7.5,3.0962)(-5,2.3219)(-2.5,1.7412)(0,1.3057)(5,0.73425)(10,0.4129)(15,0.23219)(20,0.13057)(25,0.073425)(30,0.041289)(35,0.023219)(40,0.013057)
				};
				
				\addplot[
				color=red,
				]
				coordinates {
					(-10,0.28343)(-7.5,0.28347)(-5,0.28346)(-2.5,0.28345)(0,0.28344)(5,0.28344)(10,0.28344)(15,0.28345)(20,0.28344)(25,0.28344)(30,0.28348)(35,0.28342)(40,0.28342)
				};
				
				\legend{LB ($\beta_{\text{min}} = 0.7$), MML ($\beta_{\text{min}} = 0.7$), MCRB ($\beta_{\text{min}} = 0.7$), Bias ($\beta_{\text{min}} = 0.7$),LB ($\beta_{\text{min}} = 0.3$), MML ($\beta_{\text{min}} = 0.3$), MCRB ($\beta_{\text{min}} = 0.3$),Bias ($\beta_{\text{min}} = 0.3$) }
				
			\end{semilogyaxis}
		\end{tikzpicture}
		\caption{MML, LB, MCRB, and bias term versus SNR (dB) for $\beta_{\text{min}} \in \{ 0.3,0.7\}$ when the UE distance is $5$ meters.}
		\label{fig:3}
	\end{figure}
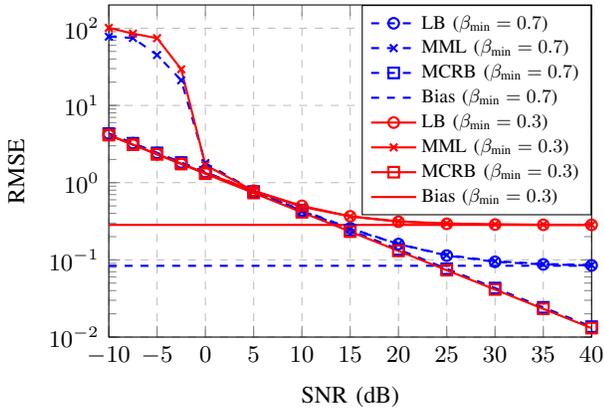


	\section{Concluding Remarks}
	
	We have considered the problem of RIS-aided near-field localization in the presence of model misspecification accounting for mismatch between ideal and realistic RIS amplitude models. In particular, we have focused on a scenario in which the exact model for the amplitudes of the RIS elements is unknown and the amplitudes are assumed to be constant (unity). 
	Based on the realistic amplitude model given in \cite{abeywickrama2020intelligent}, we have first derived the theoretical performance bounds for estimating the UE position  when the belief and the reality of the RIS amplitude models do not match. We have observed that the amplitude knowledge becomes crucial in the high SNR regime. Moreover, when the true model is known, the value of $\beta_{\text{min}}$ does not affect the theoretical bounds notably. Lastly, we have implemented the MML estimator for the considered problem, and observed that the MML estimator can achieve the LB in the high SNR regime, as expected. 
	\section*{Acknowledgment}
	This work was supported, in part, by the EU H2020 RISE-6G project under grant 101017011 and by the MSCA-IF grant
    888913 (OTFS-RADCOM).
	\appendices
	\section{Proof of Lemma 1} \label{sec:AppA}
	Based on the definition of the KL divergence and the system model in Section~\ref{sec:System}, \eqref{eq:eta0} can be expressed as
	\begin{align}
		\bet_0 &= \argmin_{\bet\in\mathbb{R}^5} \int p(\by|\bbet) \ln \left( \frac{ p(\by|\bbet)}{ \tilde{p}(\by|\bet)} \right) \, \text{d}\by  \\
		& =\argmin_{\bet\in\mathbb{R}^5} -\int p(\by|\bbet) \ln  \tilde{p}(\by|\bet)  \, \text{d}\by  \\
		& = \argmin_{\bet\in\mathbb{R}^5} \int p(\by|\bbet)  \norm{\by-\tbmu(\bet)}^2   \, \text{d}\by  \label{eq:KLlast}
	\end{align}
	where the second equality is due to the independence of $p(\by\mid\bbet)$ from $\bet$, and the last equality is obtained from \eqref{eq:assumedpdf}. Then, it can be shown that the following equations hold:
	\begin{align}	
		&\int p(\by|\bbet)  \norm{\by-\tbmu(\bet)}^2  \, d\by  = \sum_{t=1}^{T}\int p(\by|\bbet)  \abs{y_t -\tmu_t(\bet)}^2 \, \text{d}\by \nonumber \\
		& = \sum_{t=1}^{T} \underbrace{\left (\prod_{t'\neq t} \int p(y_{t'}|\bbet) \, \text{d}y_{t'}\right)}_{ = 1} \left(\int p(y_t|\bbet) \abs{y_t -\tmu_t(\bet)}^2 \, \text{d}y_t\right) \nonumber \\
		& = \sum_{t=1}^{T}\int p(y_t|\bbet) \abs{y_t -\tmu_t(\bet)}^2 \, \text{d} y_t. \label{eq:chainlast}
	\end{align}
	We now introduce $\epsilon_t(\bet)=\mu_t(\bbet)-\tmu_t(\bet)$, so that  $\abs{y_t -\tmu_t(\bet)}^2=\abs{y_t -\mu_t(\bbet) + \epsilon_t(\bet) }^2$, and the integral expression in 
	\eqref{eq:chainlast} can be manipulated as follows:
	\begin{align}\nonumber
		&\int p(y_t|\bbet) \abs{y_t -\tmu_t(\bet)}^2 \, \text{d} y_t   \\\nonumber
		&= \int p(y_t|\bbet) \abs{y_t-\mu_t(\bbet)}^2\, \text{d} y_t +
		\abs{\epsilon_t(\bet)}^2\int p(y_t|\bbet)  \, \text{d}y_t  \\\label{eq:derivLem3}
		&  + 2 \int p(y_t|\bbet)\Re\left((y_t-\mu_t(\bbet))^{*}\epsilon_t(\bet)\right) \, \text{d} y_t. 
	\end{align}
	Since $y_t\sim\mathcal{CN}(\mu_t(\bbet),N_0)$ and $\int p(y_t|\bbet)  \, \text{d}y_t = 1$, the expression in \eqref{eq:derivLem3} can be simplified as
	\begin{equation}
		\int p(y_t|\bbet) \abs{y_t -\tmu_t(\bet)}^2 \, \text{d} y_t  = N_0 +  \abs{\epsilon_t(\bet)}^2. \label{eq:integraleq}
	\end{equation}
	By combining \eqref{eq:KLlast}, \eqref{eq:chainlast} and \eqref{eq:integraleq}, we finally obtain that
	\begin{align*}
		\bar{\bet} &= \argmin_{\bet\in\mathbb{R}^5} \sum_{t=1}^{T} \left(N_0 + \abs{\epsilon_t(\bet)}^2\right)
		 = \argmin_{\bet\in\mathbb{R}^5} \sum_{t=1}^{T} \abs{\epsilon_t(\bet)}^2,
	\end{align*}
	which completes the proof.
	\section{Derivation of Entries in the MCRB} \label{sec:AppB}
Let $\bet$ be given by $\bet = [\alpha_r\, \alpha_i\, \rmx\, \rmy\, \rmz]^{\intercal}$. Also, define $\bp\triangleq [\rmx\, \rmy\, \rmz]^{\intercal}$, $b_m \triangleq [\bb(\bp)]_m$, and $\alpha \triangleq \alpha_r + j \alpha_i$. We also introduce $\bu=\frac{\bp-\bp_{\text{RIS}}}{\norm{\bp-\bp_{\text{RIS}}}}$and for any $1\leq m\leq M$, $\bu_m=\frac{\bp-\bp_m}{\norm{\bp-\bp_m}}$, where $\bu = [u_x\, u_y \, u_z]^{\intercal}$ and $\bu_m = [u_{m,x} \, u_{m,y} \, u_{m,z}]^{\intercal}$.Then, the first and second derivatives of $\tmu_t(\bet)$ with respect to $\bet$ are given as follows:
\begin{align*}
	\frac{\partial\tmu_t(\bet) }{\partial \alpha_r} &= \sum_{m=1}^{M} b_m \tilde{w}_{t,m} s_t, \, \frac{\partial\tmu_t(\bet) }{\partial \alpha_i} = j\sum_{m=1}^{M} b_m \tilde{w}_{t,m} s_t.
\end{align*}
For $\nu \in\{\rmx\, \rmy\, \rmz\}$, we can write
\begin{align*}
	\frac{\partial\tmu_t(\bet) }{\partial \nu} &= -j \frac{2\pi}{\lambda}\alpha \sum_{m=1}^{M} b_m  \left(u_{m,\nu}-u_{x}\right) \tilde{w}_{t,m} s_t, 
\end{align*}
\begin{align*}
	\frac{\partial^2\tmu_t(\bet) }{\partial \alpha_r \partial \nu} &=-j \frac{2\pi}{\lambda} \sum_{m=1}^{M} b_m  \left(u_{m,\nu}-u_\nu\right) \tilde{w}_{t,m} s_t,
\end{align*}
\begin{align*}
	\frac{\partial^2\tmu_t(\bet) }{\partial \alpha_i \partial \nu} &= j 	\frac{\partial^2\tmu_t(\bet) }{\partial \alpha_r \partial \nu},
\end{align*}
\begin{align*}
	&\frac{\partial^2 \tmu_t(\bet) }{\partial \nu \partial \nu} = -\alpha \frac{4\pi^2}{\lambda^2}\sum_{m=1}^{M} b_m \left(u_{m,\nu}-u_\nu\right)^2  \tilde{w}_{t,m} s_t \\
	&-j\frac{2\pi}{\lambda} \alpha\sum_{m=1}^{M} b_m \left(\frac{1-u^2_{m,\nu}}{\norm{\bp-\bp_m}} -\frac{1-u^2_{\nu}}{\norm{\bp-\bp_{\text{RIS}}}}\right)\tilde{w}_{t,m} s_t.
\end{align*}
Moreover, if $\nu_1, \nu_2\in\{\rmx\, \rmy\, \rmz\}$ and they correspond to different coordinates, it is possible express that 
\begin{align*}
	&\frac{\partial^2 \tmu_t(\bet) }{\partial \nu_1\partial \nu_2} = -\alpha \frac{4\pi^2}{\lambda^2}\sum_{m=1}^{M} b_m  \left(u_{m,\nu_1}-u_{\nu_1}\right)\left(u_{m,\nu_2}-u_{\nu_2}\right) \tilde{w}_{t,m} s_t  \\
	&+j\frac{2\pi}{\lambda} \alpha\sum_{m=1}^{M} b_m \left(\frac{u_{m,\nu_1} u_{m,\nu_2}}{\norm{\bp-\bp_m}}-\frac{u_{\nu_1} u_{\nu_2}}{\norm{\bp-\bp_\text{RIS}}}\right)\tilde{w}_{t,m} s_t.
\end{align*}

	\bibliographystyle{IEEEtran}
	\bibliography{bibfile}
	
\end{document}